\documentclass[11pt]{amsart}
\usepackage{amssymb, amsmath}
\usepackage{graphicx}
\usepackage{cite}

\usepackage{subfig, tikz}
\usetikzlibrary{decorations.pathmorphing}
\usepackage{graphics}

\newcommand{\td}{\text{d}}

\newcommand{\R}{\mathbb{R}}
\newcommand{\C}{\mathbb{C}}

\usepackage{bm}
\usepackage{enumerate}
\theoremstyle{plain}
\newtheorem{theorem}{Theorem}[section]
\newtheorem{lemma}[theorem]{Lemma}
\newtheorem{prop}[theorem]{Proposition}

\newtheorem{proposition}[theorem]{Proposition}
\theoremstyle{definition}
\newtheorem{remark}[theorem]{Remark}

\newenvironment{myproof}[2] {\paragraph{\it Proof of {#1} {#2}.}}{\hfill$\square$}

\def\be{\begin{equation}}
\def\ee{\end{equation}}
\def\bea{\begin{eqnarray}}
\def\eea{\end{eqnarray}}

\def\xb{X^\flat}
\def\z{\zeta}
\def\zb{\bar{\zeta}}
\def\wb{\bar{w}}
\def\p{\partial}
\def\ov{\overline}

\newcounter{mnotecount}[section]

\renewcommand{\themnotecount}{\thesection.\arabic{mnotecount}}

\newcommand{\mnote}[1]
{\protect{\stepcounter{mnotecount}}$^{\mbox{\footnotesize
$
\bullet$\themnotecount}}$ \marginpar{
\raggedright\tiny\em
$\!\!\!\!\!\!\,\bullet$\themnotecount: #1} }

\begin{document}
\title[]{Intrinsic rigidity of extremal horizons}
\author{Maciej Dunajski}
\address{Department of Applied Mathematics and Theoretical Physics\\ 
University of Cambridge\\ Wilberforce Road, Cambridge CB3 0WA, UK.}
\email{m.dunajski@damtp.cam.ac.uk}
\author{James Lucietti}
\address{
School of Mathematics and Maxwell Institute for Mathematical Sciences\\ University of Edinburgh\\
King’s Buildings, Edinburgh, EH9 3JZ, UK.}
\email{j.lucietti@ed.ac.uk}


\begin{abstract}
We prove that the intrinsic geometry of compact cross-sections of any vacuum extremal horizon must admit a Killing vector field. 
If the cross--sections
are two--dimensional spheres, this implies that the most general solution is the extremal Kerr horizon and completes the classification of the associated near-horizon geometries.  The same results  hold with a cosmological constant. Furthermore, we also deduce that any non-trivial vacuum near-horizon geometry, with a non-positive cosmological constant, must have a Lie algebra of Killing vector fields that contains  $\mathfrak{sl}(2)\times \mathfrak{u}(1)$  in all dimensions under no symmetry assumptions.  We also show that,  if the cross-sections are two-dimensional, the horizon Einstein equation is equivalent to a single fourth order PDE for the K\"ahler potential, and that this equation is explicitly solvable on the sphere if the corresponding metric admits a Killing vector. 
\end{abstract}

\maketitle

\section{Introduction}

Hawking's rigidity theorem states that, under certain global assumptions, stationary (analytic) spacetimes containing a rotating black hole must be axisymmetric~\cite{Hawking:1971vc, Chrusciel:2008js} (see also~\cite{Chrusciel:2023onh}).  This theorem has been generalised to higher-dimensional spacetimes and guarantees that a stationary rotating black hole spacetime possesses (at least) one axial Killing field~\cite{Hollands:2006rj, Moncrief:2008mr}.  While these theorems were first proven for non-extremal black holes, they have been extended to extremal black holes in four-dimensions and in higher-dimensions under a certain genericity assumption~\cite{Hollands:2008wn}.   
The basic strategy behind these proofs is to demonstrate that the event horizon is a Killing horizon, by finding a preferred foliation of a neighbourhood of the horizon by codimension-2 cross-sections. Therefore, although these theorems imply that the horizon is axisymmetric, the proofs draw upon properties of the spacetime extrinsic to the horizon.

On the other hand, as is well known, the Einstein equations for an $(n+2)$-dimensional spacetime that contains an extremal Killing horizon imply that the intrinsic geometry of the horizon decouples from the extrinsic geometry. The  geometry of such a horizon induced on a cross-section is described by an $n$-dimensional  Riemannian manifold $(M,g)$ and  a vector field $X \in \mathfrak{X}(M)$ that satisfy the horizon Einstein equation
\be
\label{eq:HEE}
\mbox{Ric}(g)=\frac{1}{2}X^\flat\otimes X^\flat
-\frac{1}{2}{\mathcal L}_{X} g+\lambda g  \; ,
\ee
where $\text{Ric}(g)$ is the Ricci tensor of $g$, ${\mathcal L}_X$ is the Lie derivative,
the one--form $X^\flat$ is $g$--dual to $X$ with respect to the metric $g$ and $\lambda$ is the cosmological constant~\cite{Kunduri:2013gce}. The same equation also arises for extremal isolated horizons~\cite{Lewandowski:2002ua}.  Indeed, this equation is satisfied by the intrinsic geometry of any extremal horizon regardless of whether it arises from a black hole spacetime.  It is also equivalent to the Einstein equations for the near-horizon geometry, which is an associated spacetime that arises from a near-horizon limiting procedure on the parent spacetime~\cite{Kunduri:2013gce} (see also~\cite{horowitz}). In this paper we shall be interested in the intrinsic geometry of $n$-dimensional  compact Riemannian manifolds $(M, g)$ that satisfy the associated elliptic system (\ref{eq:HEE}).  We refer to a solution as {\it non-trivial} if the vector field $X$ is non-identically vanishing (otherwise $(M,g)$ is Einstein). In particular, we consider whether there is  a purely intrinsic version of the rigidity theorem for extremal horizons, in other words, do all non-trivial solutions to (\ref{eq:HEE}) admit a Killing vector field?   

Numerous classification results have been obtained for the equation (\ref{eq:HEE}) under a variety of symmetry assumptions~\cite{Kunduri:2013gce}. For instance, solutions that obey $\td X^\flat=0$, which are equivalent to the near-horizon geometry being static, must be trivial if $n=2$ (four spacetime dimensions), or either trivial or the product of a circle with an Einstein metric if $\lambda \leq 0$~\cite{Chrusciel:2005pa, Bahuaud:2022iao, Wylie:2023pwf}.  Perhaps the most notable classification result states that for $n=2$ dimensions  the general solution  to \eqref{eq:nhE} on $S^2$ that admits an axial Killing field (that preserves $X$) is isometric to the horizon geometry of the extremal Kerr black hole~\cite{Lewandowski:2002ua, Kunduri:2008rs} (see also~\cite{Haj}). Explicitly, for $\lambda=0$, this geometry is
\begin{align}
\label{kerr_exact}
&g= \frac{a^2(1+x^2) \td x^2}{1-x^2} + \frac{4a^2(1-x^2) \td \phi^2}{1+x^2} , \\  &\xb= \frac{K^\flat -\td \Gamma}{\Gamma}, \quad \Gamma=\frac{1+x^2}{2}, \quad K= \frac{1}{2a^2} \partial_\phi\nonumber
\end{align}
where $a>0$ is a constant (the rotation parameter), $-1\leq x\leq 1$ and $\partial_\phi$ is a $2\pi$-periodic axial Killing vector field. It has been an outstanding open problem to determine whether  spherical topology of $M$ together with the equations (\ref{eq:HEE}) imply the existence of a Killing vector field. This would in turn imply 
that if $M=S^2$, then all solutions to (\ref{eq:HEE}) arise from some extremal Kerr solution.  Until now, the closest result in this direction has been established by Chru\'sciel, Szybka and Tod \cite{CST} who,
building on an earlier work \cite{JK}, proved the uniqueness of the Kerr horizon (\ref{kerr_exact})
in its neighbourhood in the space of all solutions of the horizon equation (\ref{eq:HEE}).

 In this paper we show that any non-trivial solution to (\ref{eq:HEE}) on a compact surface must in fact admit a Killing field.  
 We will actually prove much more: the existence of a Killing field holds in any dimension and also with a cosmological
 constant. Our main result is the following.
 \begin{theorem}
 \label{maintheo0}
 Let $(M, g)$ be an $n$-dimensional compact, orientable\footnote{If $M$ is not orientable then the pair $(g, X^\flat)$ can be pulled back to the oriented covering space $\hat{M}$ of $M$. The Killing vector field on $\hat{M}$ resulting from Theorem \ref{maintheo0} can then be pushed forward to a Killing vector of  $(M, g)$.}, Riemannian
 manifold without boundary admitting  a non-gradient vector field $X$   such that (\ref{eq:HEE}) holds. Then there exists a Killing vector field $K$ of $(M, g)$  such that $K^\flat=\Gamma X^\flat+\td \Gamma$ where $\Gamma$ is a smooth positive function. Furthermore,  if either (i) $\lambda\leq 0$, or (ii) $n=2$ and $\lambda$ is arbitrary, then $[K, X]=0$. 
 \end{theorem}
 This result is complementary to the aforementioned classification of static horizons which correspond to $X^\flat$ closed.
The proof of Theorem \ref{maintheo0} makes use of the existence of a principal eigenvector of a certain second order elliptic operator \cite{Andersson:2007fh, Lucietti:2012sf}, together with a remarkable tensor identity (\ref{eq:identity})
 which we shall establish in Proposition \ref{prop}. The principal eigenvector will give a candidate for a Killing vector $K$, and the identity will allow to reduce the $g$--norm 
 $|{\mathcal L}_X g|^2$ to a total divergence.
 Thus its integral over a closed manifold $M$
 vanishes which implies that $K$ is a solution to the Killing equations.   The identity then implies that the vector field $X$ is also invariant under the Killing field provided $\lambda \leq 0$, thus  establishing part (i) of the theorem.  This latter condition is significant because it implies that $K$ extends to a Killing field of the associated near-horizon geometry.
 It is worth emphasising that the validity of the identity (\ref{eq:identity}) is based on several cancellations which depend crucially on the precise numerical factor appearing in the quadratic in $X$ term of the horizon equation (\ref{eq:HEE}). For $n=2$  a more detailed analysis of the horizon equation establishes part (ii) of the theorem.

We will also show that if the dimension $n=2$, the system (\ref{eq:HEE}) reduces to a single fourth--order non--linear PDE
for the K\"ahler potential of $g$ with the vector field $X$ determined purely in terms of this potential.
This equation can introduced in elementary terms and using only
the vector calculus on the plane. 
Let 
$f:\R^2\rightarrow \R$, and let $\nabla$ and $\Delta=\nabla\cdot\nabla$ 
be the nabla operator and the Laplacian of the flat Euclidean metric 
respectively. Then the PDE reads
\be
\label{vector_pde}
\frac{1}{2}{\Delta}^2 f-\frac{1}{2}\frac{|\nabla(\Delta f)|^2}{\Delta f}+
2\frac{(\Delta f)^3}{|\nabla f|^2}+\Delta f\nabla f\cdot
\nabla\Big( \frac{\Delta f}{|\nabla f|^2} \Big)+\lambda(\Delta f)^2=0.
\ee
We prove that (\ref{vector_pde})
is (locally) equivalent to (\ref{eq:HEE}) where the function $f$ is the K\"ahler
potential for $g$.  While establishing, or ruling out the integrability of (\ref{vector_pde}) is an interesting open problem,
there is  some evidence for its solvability: If $g$ admits 
a Killing vector field (as guaranteed by Theorem \ref{maintheo0} for compact $M$), then a reciprocal 
change of variables  can be 
used to linearise (\ref{vector_pde}), and find all solutions explicitly on $S^2$, therefore providing an alternate derivation of the extremal Kerr horizon. 

Theorem \ref{maintheo0} has a number of corollaries.  The most important one, alluded to above, is as follows.

\begin{theorem} \label{cor1}The extremal Kerr horizon (possibly with cosmological constant) is the unique solution to \eqref{eq:HEE} on $M=S^2$.
\end{theorem} 

For $n=2$, the Gauss-Bonnet theorem implies that the only possible non-trivial solutions to (\ref{eq:HEE}) for $\lambda \geq 0$ are on $M=S^2$.  For $\lambda<0$, although higher genus surfaces are possible, it has also been shown that the only solution is trivial~\cite{Dobkowski-Rylko:2018nti}. Therefore, together with the classification of solutions for which $X$ is a gradient field~\cite{Chrusciel:2005pa}, our result completes the classification of vacuum extremal horizons with  two-dimensional compact cross-sections for arbitrary cosmological constant.

Another important corollary to our main theorem is a symmetry enhancement result for the corresponding near-horizon geometries.
\begin{theorem} \label{cor2}
Any non-trivial vacuum near-horizon geometry with $\lambda\leq 0$ and a compact orientable cross-section has a Lie algebra of Killing vector fields that contains $\mathfrak{sl}(2)\times \mathfrak{u}(1)$.
\end{theorem}

This result is a substantial generalisation of the  near-horizon symmetry enhancement theorem first proven under the additional assumption that $(M, g, X)$ admits a $U(1)^{n-1}$-axial symmetry~\cite{Kunduri:2007vf, Lucietti:2012sa}.
The proof relies on the existence of a first integral to (\ref{eq:HEE}) under the assumptions of Theorem \ref{maintheo0}.  This is enough to show that one can write the near-horizon geometry as a warped fibration of $M$ over (a patch of) $2$-dimensional anti de Sitter space (AdS$_2$) and that this structure inherits the Killing fields of the AdS$_2$ base space. 
This result generalises to $\lambda>0$ if the cross-sections are two-dimensional, although in that case we already know the near-horizon geometry is that of extremal Kerr de Sitter from Theorem \ref{cor1}.  For $\lambda>0$ and $n>2$, we have not been able to prove that the vector field $X$ is invariant under the Killing field  $K$ of Theorem \ref{maintheo0}. This would be required to establish a version of Theorem \ref{cor2} for $\lambda>0$ in all dimensions.   

\subsection*{Note added.}
{\em Since this work was completed, the invariance condition $[K, X]=0$ has been proven for any $\lambda$ and in all dimensions~\cite{Colling:2024usk}. Moreover, Theorem \ref{maintheo0} has been extended to Einstein--Maxwell
theory (thus proving the intrinsic rigidity of extremal Kerr--Newman)~\cite{Colling:2024txz}. The tensor identity (7)
has been strengthened~\cite{Kaminski:2024tob}.}

\section{Existence of Killing field}
\label{section1}
In this section we will establish our main result. We first prove a general identity for solutions to the system (\ref{eq:HEE}), which allows one to deduce the existence of a Killing field. We then show the system admits a first integral, which implies the corresponding near-horizon geometry possesses an enhanced isometry group.

\subsection{General identity}
Consider the near--horizon equation (\ref{eq:HEE})
written in a tensor form
\begin{equation}
R_{ab}= \frac{1}{2} X_a X_b - \nabla_{(a} X_{b)}+\lambda g_{ab} \; . \label{eq:nhE}
\end{equation}
We denote the $g$-norm of any tensor by $| \cdot |$ and the Laplacian by $\Delta := \nabla^a \nabla_a$. 
Inspired by the Kerr solution (\ref{kerr_exact}), for any smooth positive function $\Gamma:M \to \mathbb{R}^+$ we introduce a vector field $K$ defined by 
\begin{equation}
K_a:= \Gamma X_a+ (\td \Gamma)_a \; .
 \label{eq:Xdecomp}
\end{equation}
Of course, there is no loss in generality in doing this since we have not specified $\Gamma$. Shortly, we will make a specific choice of $\Gamma$, but for now there are no restrictions on it.  Inverting gives $X_a= (K_a- \nabla_a\Gamma)/\Gamma$ and the horizon equation (\ref{eq:nhE}) then becomes
\begin{equation}
R_{ab}= \frac{K_a K_b }{2\Gamma^2} -\frac{(\nabla_a \Gamma)( \nabla_b\Gamma)}{2 \Gamma^2}- \frac{1}{\Gamma}\nabla_{(a} K_{b)}+\frac{1}{\Gamma} \nabla_a\nabla_b \Gamma+\lambda g_{ab}  \; . \label{eq:nhE2}
\end{equation}
We will now state our general identity.

\begin{proposition}\label{prop}
For any solution to (\ref{eq:nhE2}) the following identity holds
\begin{align}
\nabla_{(a} K_{b)} \nabla^a K^b &= \nabla^a \left( K^b \nabla_{(a} K_{b)} -\tfrac{1}{2}K_a\Delta\Gamma- \tfrac{1}{2} K_a \nabla_b K^b- \lambda \Gamma K_a \right)  \label{eq:identity} \\ \nonumber
&+ \nabla_a K^a \left(  -\frac{1}{2 \Gamma} |K|^2 +\frac{1}{2}\Delta\Gamma +\frac{1}{2} \nabla_b K^b + \frac{1}{2\Gamma}K^b\nabla_b\Gamma +\lambda \Gamma\right).
\end{align}
\end{proposition}
\begin{proof}
First note that we can write
\begin{equation}
\nabla_{(a} K_{b)} \nabla^a K^b = \nabla^a \left( K^b \nabla_{(a} K_{b)} \right)  - K^b \nabla^a \nabla_{(a} K_{b)}  \; .   \label{eq:Killingeqsq}
\end{equation}
To compute the second term we use the contracted Bianchi identity $\nabla^a(R_{ab}-\tfrac{1}{2}R g_{ab})=0$ applied to (\ref{eq:nhE2})  to obtain an expression for $ \nabla^a \nabla_{(a} K_{b)}$. The result is rather  messy and given by
\begin{align}
 \nabla^a \nabla_{(a} K_{b)}  &= \frac{1}{\Gamma} (\nabla^a\Gamma )\nabla_{(a} K_{b)} +\frac{1}{2\Gamma} K_b \nabla_c K^c +\frac{1}{2\Gamma} K^a \nabla_a K_b  - \frac{1}{\Gamma^2} K_b K^c \nabla_c\Gamma  \label{eq:calc1}  \\  \nonumber
 &- \frac{1}{2\Gamma} (\Delta \Gamma)  \nabla_b \Gamma -\frac{3}{2\Gamma} (\nabla^a\Gamma)\nabla_a \nabla_b\Gamma  + \frac{1}{\Gamma^2} | \nabla \Gamma|^2 \nabla_b \Gamma +\nabla^a \nabla_b \nabla_a \Gamma  \\ \nonumber
 & -\tfrac{1}{2} \Gamma \nabla_b \left( \frac{1}{2\Gamma^2} |K|^2 - \frac{1}{2\Gamma^2} |\nabla \Gamma|^2-\frac{1}{\Gamma} \nabla_c K^c  +\frac{1}{\Gamma} \Delta \Gamma \right)  \; .
\end{align}
The triple derivative of $\Gamma$ term can be evaluated using the identity $[\nabla_a, \nabla_b ] V^a= R_{ab} V^a$ and (\ref{eq:nhE2}) again  (this in particular cancels the $(\nabla^a\Gamma )\nabla_{(a} K_{b)}$ term in (\ref{eq:calc1})).  Next, contracting (\ref{eq:calc1}) with $K^b$ one finds that all terms can be written either as divergences, or proportional to $K^b\nabla_b \Gamma$, or proportional to $\nabla_b K^b$. Remarkably many cancellations occur resulting in the claimed identity (\ref{eq:identity}). Details of this calculation are given in Appendix \ref{app:mainidentity}. 
\end{proof}

We will now make a specific choice of the function $\Gamma$, which has been previously used to establish instability of extremal horizons~\cite{Lucietti:2012sf} (see also \cite{DGS}
for an application to  3-dimensional horizons).
We will repeat the proof in the present notation for
convenience (see also Appendix of \cite{tod} for an alternate
proof). Note that this result does not require (\ref{eq:nhE}) to hold. In fact, the next two results (Lemma \ref{lemma} and Theorem \ref{th:main}) involve a  global integration so we shall assume that $M$ is orientable, and chose an orientation given by a volume form parallel with respect to the metric $g$.
\begin{lemma}\label{lemma}
Given any vector field $X$ on a compact, oriented, Riemannian manifold $(M,g)$ there exists (a unique up to scale) smooth function $\Gamma>0$ such that $\nabla_a K^a=0$ where $K$ is defined by \eqref{eq:Xdecomp}.
\end{lemma}
\begin{proof}
Define the second order elliptic operator on $(M, g)$ by 
\begin{equation}
L \psi := -\nabla^a(\nabla_a \psi+  X_a \psi)
\end{equation}
where $\psi \in C^\infty (M)$.   According to \cite[Lemma 4.1]{Andersson:2007fh}, for $M$ compact, there exists a principal eigenvalue $\mu$ which is real and less than or equal to the real part of any other eigenvalue, whose associated eigenfunction $\psi$ is everywhere positive and unique up to scale.  Thus $L\psi= \mu \psi$ and integrating this over $M$  implies $\mu=0$ by Stokes' theorem.  Therefore, we may take $\Gamma$ to be (proportional) to the principal eigenfunction which satisfies $L\Gamma=0$.  In terms of $K$ defined by (\ref{eq:Xdecomp}) this proves the claim.
\end{proof}

We are now ready to establish our main result which proves Theorem \ref{maintheo0} (i).
\begin{theorem} \label{th:main}
Let $(M, g)$  be a compact, oriented, Riemannian manifold admitting a non-gradient vector field $X$ satisfying  (\ref{eq:nhE}). Then, there exists a (unique up to scale) positive function $\Gamma$ such that $K$ defined by (\ref{eq:Xdecomp}) is Killing vector field of $(M, g)$. Furthermore, if $\lambda \leq 0$, then $[K, X]=0$.
\end{theorem}

\begin{proof}
By Lemma \ref{lemma} there exists a function $\Gamma>0$ such that $\nabla_a K^a=0$ where $K_a := \Gamma X_a+ \nabla_a\Gamma$ (i.e. defined by \eqref{eq:Xdecomp}).   Now, identity (\ref{eq:identity}) simplifies drastically to
\begin{equation}
\nabla_{(a} K_{b)} \nabla^a K^b = \nabla^a \left( K^b \nabla_{(a} K_{b)} -\tfrac{1}{2}K_a\Delta \Gamma- \lambda \Gamma K_a \right)  \; .  \label{eq:id2}
\end{equation}
Integrating this over $M$ and using Stokes' theorem implies that $\nabla_{(a} K_{b)}=0$ everywhere on $M$.  If $K$ vanishes identically then $X_a=- \nabla_a(\log \Gamma)$ is a gradient. Thus, under the stated assumptions in Theorem \ref{maintheo0},  $K$ must not vanish identically and hence defines a Killing vector field of $(M, g)$.

The identity (\ref{eq:id2}) now simplifies to
\begin{equation}
\label{eq:LKG}
\mathcal{L}_K\left( \tfrac{1}{2}\Delta \Gamma + \lambda \Gamma \right)=0  \qquad \implies \qquad \Delta (\mathcal{L}_K\Gamma) + 2\lambda \mathcal{L}_K \Gamma =0 
\end{equation}
where the second equation follows since $K$ is a Killing field.  For $\lambda=0$ this implies $ \mathcal{L}_K \Gamma $ is constant and since $ \mathcal{L}_K \Gamma =\nabla^a(K_a \Gamma)$ the constant must vanish. For $\lambda<0$ multiply the second equation in (\ref{eq:LKG}) by $\mathcal{L}_K \Gamma$ and integration by parts shows that $\mathcal{L}_K \Gamma=0$.  Hence $[K, X]=0$.
\end{proof}

The final part of the above proof, which shows that $X$ is invariant under $K$, does not work if $\lambda>0$.   However, it shows that even if $\lambda>0$ there is a Killing field $K$ of the form (\ref{eq:Xdecomp}). Using this, in dimension $n=2$ we can obtain the following result, which establishes Theorem \ref{maintheo0} (ii).

\begin{theorem}
\label{prop2d}
 Let $(M, g)$ be a $2$-dimensional compact, oriented, Riemannian
manifold  with a non-gradient vector field $X$ such that the horizon equation  (\ref{eq:nhE}) holds for some $\lambda$ (positive, zero, or negative).  Then $[K, X]=0$ where $K$ is the Killing field in Theorem \ref{th:main}.
 \end{theorem}
 
 \begin{proof}
 By the first part of the proof of Theorem \ref{th:main}, if $X$ is non-gradient, there exists a function $\Gamma>0$ and Killing vector field $K$ defined by (\ref{eq:Xdecomp}).  With this choice, multiply equation (\ref{eq:nhE2}) by $\Gamma$ and apply $\epsilon^{ac}\nabla_c$ where $\epsilon_{ab}$ is the volume form of $(M, g)$. Then, commute derivatives of the $\nabla^3\Gamma$ term using $[\nabla_a, \nabla_b]V_c=- R^d_{~cab} V_d$ and use again (\ref{eq:nhE2}) to eliminate the 2nd derivatives of $\Gamma$.  Next, multiply by $\epsilon^{b}_{~d}$ and simplify using $\epsilon^{bd}\epsilon^{ac}= g^{ba}g^{dc}- g^{bc}g^{da}$ and the fact $K$ is a Killing  field. This gives the identity
 \begin{equation}
 \left( R- \frac{2\lambda}{3} \right) \nabla_d \Gamma + \frac{2\Gamma}{3} \nabla_d R =2 K^b\nabla_{[d}X_{b]} \;.\label{eq:id2d}   
 \end{equation}
 More details of the derivation of (\ref{eq:id2d})  are given in Appendix \ref{app:id2d}.
 Now contracting  (\ref{eq:id2d})  with $K^d$ gives  
  \begin{equation}
  \left( R- \frac{2\lambda}{3} \right) \mathcal{L}_K\Gamma=0  \;. 
 \end{equation} 
Therefore, either $R=2\lambda/3$ identically and the space is constant curvature, or $\mathcal{L}_K\Gamma=0$.  In the latter case  it follows that $[K,X]=0$ as required.

To investigate the case $R=2\lambda/3$, prolong equation (\ref{eq:nhE}) for $n=2$ to obtain
\begin{equation}
\nabla_a X_b = \frac{1}{2} X_ a X_b + \left( \lambda - \tfrac{1}{2} R \right) g_{ab} + \frac{1}{2}\Omega \epsilon_{ab}   \label{eq:prolong}
\end{equation}
where $\Omega$ is the function defined by $\Omega=\star \td X^\flat$ and $\star$ is the Hodge star operator.  Now by applying $\epsilon^{ca}\nabla_c$ to \eqref{eq:prolong} and commuting derivatives one can show
\begin{equation}
\left(\td - \frac{3}{2} X^\flat \right)\Omega - \star \left(\td - \frac{3}{2} X^\flat \right) \left( R- \frac{2\lambda}{3} \right)=0  \; .
\end{equation}
This expression was also found in~\cite{Dobkowski-Rylko:2018nti}.
Therefore, if $R=2\lambda/3$ we deduce  that $(\td - \tfrac{3}{2}X^\flat)\Omega=0$. Taking the exterior derivative of this expression and substituting back in gives $\Omega\wedge \star \Omega=0$ and hence $\td X^\flat=0$.  But if $\td X^\flat=0$ then $X$ must vanish identically as shown in~\cite{Chrusciel:2005pa, Bahuaud:2022iao}.
 \end{proof}

\begin{remark}
\label{constR}
For completeness we shall present the general solution in the  constant curvature case encountered in the above proof where 
$R=2\lambda/3$ and $X$ is a gradient. If $M=S^2$  then necessarily $\lambda>0$ by the Gauss--Bonnet  
theorem. We can chose the spherical  coordinates $(\theta, \phi)$ on $M$ such that
\be
\label{non_inh1}
g=\frac{3}{\lambda}\Big(\td \theta^2+
\sin^2{\theta}\td \phi^2\Big), \quad
\xb=-2\frac{\td F}{F}
\ee
for some function $F=F(\theta, \phi)$. The $(\theta\theta)$ component of the horizon equation 
(\ref{eq:HEE}) gives $F=f_1(\phi)\sin{\theta}+f_2(\phi)\cos(\theta)$, where $f_1, f_2$ are functions of $\phi$.
The $(\theta\phi)$ component then gives $f_2=b=\mbox{const}$. Finally the $(\phi\phi)$ component gives
$f_1=a\sin{\phi}$ where one constant of integration has been absorbed into $\phi$, and $a$ is another constant. Therefore, the final solution is (\ref{non_inh1}) with
\be
\label{form_of_F}
F=a\sin{\phi}\sin{\theta}+b\cos{\theta}.
\ee
The resulting $X^{\flat}$ is singular on $S^2$ at $F=0$ which occurs for any $a,b$ (as must be the case since smooth gradient solutions on compact $M$ are trivial ~\cite{Chrusciel:2005pa, Bahuaud:2022iao}).

\vskip2pt

If the scalar curvature of $g$ is non-constant then the choice of the Killing vector $K$ is unique, up to a constant multiple of $K$. In the constant curvature case above the metric $g$ is $SO(3)$--invariant, and although $K=\p_\phi$ does not preserve $X$ unless
$a=0$, there is a choice of 
$U(1)\subset SO(3)$ generated by the Killing vector
\[
K=a(\cos{\phi}\p_\theta-\sin{\phi}\cot{\theta}\p_\phi)+b\p_\phi
\]
which keeps $X$ invariant.

The following example shows another occurrence of
this: Consider
a pair $(g=\td x^2+\td y^2, \xb=-2x^{-1}\td x)$ which solves (\ref{eq:HEE}) on $M=\mathbb{R}^2$ (this in fact is a special case of the static horizons~\cite{Chrusciel:2005pa}). The isometry group of $g$ is three--dimensional, and while $X=-2x^{-1}\partial_x$ commutes with $\p_y$, it does not commute with the other two Killing vectors $\p_x$ and $x\p_y-y\p_x$.  
\end{remark}

\begin{remark}
In fact, although not needed in the proof, $[K, X]=0$ implies $\mathcal{L}_K\Gamma=0$, if $K$ is a Killing field.   To see this, suppose $[K, X]=0$.  Then Lie deriving $L\Gamma=0$ along $K$, where $L$ is the operator in Lemma \ref{lemma}, we deduce that $L (\mathcal{L}_K \Gamma)=0$.  Hence $\mathcal{L}_K \Gamma$ must be a constant multiple of the principal eigenfunction $\Gamma$ of $L$ and integrating over $M$ implies this constant vanishes.

If one assumes $X$ is non-gradient (as  in Theorem \ref{maintheo0}) and $K$ is of the form (\ref{eq:Xdecomp}) this result follows from a local argument. If $[K,X]=0$ then Lie-deriving $K_a=\Gamma X_a +(\td \Gamma)_a$ along $K$ gives   $0=K(\Gamma) X_a + (\td (K(\Gamma)))_a$ and therefore $X$ is a gradient-field unless $K(\Gamma)=0$ identically.
\end{remark}

\subsection{Uniqueness of extremal Kerr horizon}

The most important consequence of Theorem \ref{maintheo0} is that it allows us to complete the classification of solutions to \eqref{eq:HEE} on $M=S^2$. \\
\begin{myproof}{Theorem}{\ref{cor1}}
If $X$ is non-gradient, Theorem \ref{maintheo0} shows that any solution on $M=S^2$ must admit a Killing vector field $K$ that leaves $X$ invariant.  All smooth solutions of this kind were determined in~\cite{Lewandowski:2002ua, Kunduri:2008rs}, with the result that the general solution is isometric to the extremal Kerr horizon, including the case of a cosmological constant.

If $X$ is gradient all solutions on $S^2$ must be trivial~\cite{Chrusciel:2005pa} and hence Einstein with $\lambda>0$. This corresponds to the horizon of an extremal Schwarzschild de Sitter spacetime which is a special case of the extremal Kerr de Sitter horizon.
\end{myproof}

\begin{remark}
In Section \ref{section4} we give an alternate derivation of the general axially symmetric solution on $M=S^2$ using K\"ahler geometry.
\end{remark}

\subsection{Near-horizon symmetry enhancement}
\label{section2}

We can extract a first integral of the system (\ref{eq:nhE}). Define the function
\begin{equation}
F := \frac{1}{2} | X|^2- \frac{1}{2} \nabla_a X^a+\lambda  \; .  \label{eq:Fdef}
\end{equation}
It has been shown that the contracted Bianchi identity for (\ref{eq:nhE}) is equivalent to\cite{Kunduri:2013gce, Kunduri:2008rs}
\begin{equation}
\nabla_a F- X_a F - 2 X^b\nabla_{[a} X_{b]}+ \nabla^b \nabla_{[a}X_{b]}=0  \; .   \label{eq:Bid}
\end{equation}
Now define another function $A$ (again inspired by the near-horizon geometry of Kerr) by
\begin{equation}
F=: \frac{A}{\Gamma} +\frac{|K|^2}{\Gamma^2}  \; .    \label{eq:Adef}
\end{equation}
Then, a tedious calculation, using (\ref{eq:Xdecomp}) and (\ref{eq:Adef}), shows that (\ref{eq:Bid}) can be written as
\begin{equation} 
\nabla_a A+\frac{1}{2\Gamma} \left(\nabla_a -  \frac{K_a +\nabla_a\Gamma}{\Gamma} \right) (K^b \nabla_b \Gamma) -  \left(\nabla^b - \frac{K^b}{\Gamma} \right)\nabla_{(a} K_{b)} + \left(\nabla_a+\frac{K_a-\nabla_a\Gamma}{2\Gamma} \right) \nabla_c K^c=0.   \label{eq:DA}
\end{equation}
Therefore, we obtain the following.
\begin{prop} \label{prop:A}
Under the assumptions of Theorem \ref{maintheo0} and $\lambda\leq 0$ the function $A$ defined by (\ref{eq:Adef}) is a negative constant. 
\end{prop}
\begin{proof}
Constancy of $A$ immediately follows from (\ref{eq:DA}) together with Theorem \ref{maintheo0}. Then, using the definitions of $F$ and $A$ one finds
\begin{equation}
A= -\frac{1}{2\Gamma} |K |^2 + \frac{1}{2}\Delta \Gamma + \lambda \Gamma  \; .
\end{equation}
Integrating over $M$ for $\lambda \leq 0$  shows that $A\leq 0$, and furthermore, $A=0$ if and only if $\lambda=0$, $K_a=0$  and $\Gamma$ constant (which gives trivial $X=0$).
\end{proof}

The significance of this is revealed by the following. If $(M, g)$ is an $n$-dimensional Riemannian manifold with a vector field $X\in \mathfrak{X}(M)$ and a smooth function $F$,  the near-horizon geometry is the associated $(n+2)$-dimensional spacetime $(\mathbb{R}^2\times M, \mathbf{g})$,
\begin{equation}
\mathbf{g}= r^2 F \td v^2+ 2 \td v \td r + 2 r X^\flat \td v+ g \; ,   \label{eq:nh}
\end{equation}
where the coordinates $(v, r)\in \mathbb{R}^2$. This satisfies the vacuum Einstein equations $\text{Ric}(\mathbf{g})= \lambda \mathbf{g}$ if and only if $(g, X, F)$ satisfy (\ref{eq:nhE}) and (\ref{eq:Fdef}), see e.g.~\cite{Kunduri:2013gce}. The surface $r=0$ is an extremal Killing horizon of the Killing vector $\partial_v$. A generic near-horizon geometry has a 2-dimensional non-abelian isometry group generated by $v$-translations and $(v, r)\to (\lambda^{-1} v, \lambda r)$ dilations.   We will now show that the above results imply that the near-horizon geometry has enhanced symmetry. To this end, define a new coordinate $\rho \in \mathbb{R}$ by $r  = \Gamma \rho$ in terms of which the near-horizon geometry becomes
\begin{equation}
\mathbf{g}=\Gamma [ A \rho^2 \td v^2+ 2 \td v \td \rho ]+  g+ 2\rho K^\flat \td v + |K|^2 \rho^2 \td v^2 \; ,  \label{eq:nh2}
\end{equation}
where $A$ and $K$ are defined by (\ref{eq:Xdecomp}) and (\ref{eq:Adef}).  We are ready to deduce the following. \\

\begin{myproof}{Theorem}{\ref{cor2}}  If $X$ is a gradient then it has been shown in all dimensions that if $\lambda\leq 0$ then $X$ vanishes identically~\cite{Chrusciel:2005pa, Bahuaud:2022iao}.  If $X$ is non-gradient Theorem \ref{maintheo0} applies and hence Proposition \ref{prop:A} shows that $A_0:=A<0$ is a constant, so $\gamma:=A_0 \rho^2 \td v^2+ 2 \td v \td \rho$, the 2-dimensional metric in the square brackets of (\ref{eq:nh2}), defines a constant negative curvature Lorentzian space $(\mathbb{R}^2, \gamma)$ that is isometric to an open set of  AdS$_2$ in Poincar\'e coordinates (adapted to a Killing horizon at $\rho=0$). The volume form of this space is $\text{vol}_\gamma= \td \theta$ where $\theta= \rho \td v$. Let $Y$ be a Killing vector field of $\gamma$. Then $0=\mathcal{L}_Y \td\theta=\td\mathcal{L}_Y \theta$ and hence $\mathcal{L}_Y\theta= \td f_Y$ for some function $f_Y$ on $\mathbb{R}^2$. The Killing vector fields of $\gamma$ in the $(\rho, v)$ coordinates are
\be
Y_1=\tfrac{1}{2} v^2\partial_v +(A_0^{-1}- \rho v) \partial_\rho, \qquad Y_2= v \partial_v- \rho\partial_\rho, \qquad Y_3=  \partial_v  \; .
\ee
Then, setting $f_\alpha:= f_{Y_{\alpha}}, \alpha=1, 2, 3$, we find the explicit expressions
\be
\label{eq:falpha}
f_1=A_0^{-1}v+c_1, \quad f_2=c_2, \quad f_3=c_3,
\ee
where $c_\alpha$ are constants.  These Killing fields form an $\mathfrak{sl}(2)$  Lie algebra. We now show how to extend these to Killing vector fields of the near-horizon geometry $(\mathbb{R}^2\times M, \mathbf{g})$. 

Clearly, $K$ extends to a Killing vector field of the near-horizon geometry \eqref{eq:nh} tangent to $M$, and the AdS$_2$ Killing fields $Y_\alpha$ extend to vector fields on $\mathbb{R}^2\times M$ tangent to $\mathbb{R}^2$ that Lie-derive all the data $(g, X, F)$ on $M$.
Let $h$ be the metric on the space of orbits of $K$ in $(M, g)$ so that on the open set where $K\neq 0$ we have
\[
g=h+|K|^{-2} K^\flat \otimes K^\flat, 
\]
and we can write the near-horizon geometry \eqref{eq:nh2} as
\be
\label{eq:nhalt}
{\bf{g}}=\Gamma \gamma+h+ |K|^{-2}(K^\flat+|K|^2\theta)^2   \; .
\ee
Away of the zeros of $K$ define a function $\phi:M \to \mathbb{R}$ such that $K(\phi)=1$. Then we can write $K^\flat= |K|^2 (\td \phi +\Theta)$ for some 1-form $\Theta$ on the space of orbits of $K$ (in particular, $\iota_K\Theta=0$ and $\mathcal{L}_K \Theta=0$) and hence
\be
\label{eq:nhalt2}
{\bf{g}}=\Gamma \gamma+h+ |K|^2\omega^2   \; ,
\ee
where $\omega:= \td \phi+ \Theta+\theta$. Therefore, we can regard \eqref{eq:nhalt2} as a metric on a non-trivial line  (or circle) bundle over $\mathbb{R}^2$ times the space of orbits of $K$ in $M$, with connection $\omega$. Indeed, since $K(\phi)=1$ it follows that $\iota_K\omega=1$ so $K$ is vertical.  Now, we have $\mathcal{L}_{Y_\alpha} \omega= \mathcal{L}_{Y_\alpha} \theta= \td f_\alpha$ and  $\mathcal{L}_{f_\alpha K} \omega=\td f_\alpha$  and therefore we deduce that
\be
\overline{Y}_{\alpha}=Y_\alpha-f_\alpha K, \quad \alpha =1, 2, 3,  \label{eq:Ybar}
\ee
which satisfy $\mathcal{L}_{\bar{Y}_\alpha}\omega=0$,  are Killing vector fields of the metric (\ref{eq:nhalt2}). The metric (\ref{eq:nhalt2}) is valid on the region $K\neq 0$, however, at any zero of $K$ the vector fields $\bar{Y}_\alpha=Y_\alpha$ still leave the near-horizon geometry \eqref{eq:nh2} invariant. Therefore, \eqref{eq:Ybar} are Killing vector fields of the near-horizon geometry $(\mathbb{R}^2\times M, \mathbf{g})$.

Now setting the constants $c_1=0, c_2=A_0^{-1}, c_3=0$, it is easily verified that $[\bar{Y}_\alpha, \bar{Y}_\beta ] = c_{\alpha\beta\gamma} \bar{Y}_\gamma$, where the structure constants $c_{\alpha\beta\gamma}$ are equal to those of the Lie algebra of Killing fields of AdS$_2$, that is,  $[{Y}_\alpha, {Y}_\beta ] = c_{\alpha\beta\gamma}{Y}_\gamma$.  Hence the near-horizon geometry has an $\mathfrak{sl}(2) \times \mathfrak{u}(1)$ Lie algebra of Killing fields where the $\mathfrak{sl}(2)$ is generated by $\bar{Y}_\alpha$ and the $\mathfrak{u}(1)$ by $K$.
\end{myproof}
\begin{remark}
As mentioned above, $(\mathbb{R}^2, \gamma)$ is isometric to an open set of AdS$_2$ and hence it is extendible as a spacetime to global AdS$_2$.  The Killing field $Y_1$ is not complete on $(\mathbb{R}^2, \gamma)$ and hence does not generate a 1-parameter subgroup of isometries.  This is why the above result is formulated in terms of the Lie algebra of Killing vector fields as opposed to the isometry group. However, the near-horizon geometry (\ref{eq:nhalt2}) is also extendible to a complete spacetime by changing trivialisation of the $\mathbb{R}^2$ factor of the base so that it extends to global AdS$_2$. Then, on this extended near-horizon geometry spacetime the theorem can be strengthened to say that its isometry group has an $SO(2,1)$ subgroup with $3$-dimensional orbits; if $K$ has closed orbits this gives an $SO(2,1)\times U(1)$ isometry group.  In higher dimensions $K$ need not have closed orbits; for example, for the $5$-dimensional Myers-Perry black hole (so $M$ is $3$-dimensional), the isometry group of the (extended) near-horizon geometry is $SO(2,1)\times U(1)^2$ and $K$ is a Killing field that is tangent to a non-closed orbit of $U(1)^2$~\cite{Kunduri:2007vf}.
\end{remark}

\section{K\"ahler potential in 2-dimensions}
\label{section4}
In this section we will show that if $M$ is 2-dimensional the horizon Einstein equation (\ref{eq:HEE}) reduces to a fourth order PDE for the K\"ahler potential of the metric, and use this formulation together with Theorem \ref{maintheo0} to show that  on $M=S^2$ this PDE is explicitly solvable.

\subsection{4th order PDE}
The complex structure on $M=S^2$ is defined by an atlas consisting of two open sets $U$ and $\widetilde{U}$
with holomorphic coordinates $\zeta\in \C$ and $\eta\in\C$ respectively, and the overlap holomorphic patching relation $\eta=\zeta^{-1}$ on $U\cap\widetilde{U}$. In these coordinates the round (or constant scalar curvature) metric on $M$ takes the form
\be
\label{g00}
g_0=\frac{4\td \zeta \td \zb}{(1+|\zeta|^2)^2},
\ee
and by the uniformization theorem any other metric on $M$ can be written as
the form
\be
\label{guniform}
g=e^u g_0
\ee
where $u$ is a globally defined function. In what follows we shall use a K\"ahler potential for $(M, g)$, that is,
a real valued function $f$ such that
\be
\label{kahler11}
g=4f_{\zeta\bar{\zeta}} \; \td\zeta \td\bar{\zeta}
\ee
where $f_{\z}=\p_{\z} f$ etc.

While such $f$ can not be defined globally on $M$, the $\p\ov{\p}$ lemma implies that if $\omega_0$ and $\omega$ are K\"ahler
two--forms of $g_0$ and $g$ respectively, then there exists a global function $f_1: M\rightarrow \R$ such that
\[
\omega=\omega_0+i\p\overline\p f_1, \quad\mbox{where}\quad \p=\td \zeta\otimes\partial_{\zeta}.
\]
Therefore, if we chose a local K\"ahler potential $f_0$ for the round metric $g_0$, then  we can set
\[
f=f_0+f_1,
\]
where $f_1$ is global. Adding a holomorphic function  and its conjugate (or equivalently a harmonic function) to the K\"ahler potential does not change the metric. This freedom can be partially fixed by the following argument: On the round sphere the local potential $f_0=\mbox{ln}(1+|\zeta|^2)$ yields the constant curvature metric (\ref{g00}), but blows up 
as $\zeta\rightarrow\infty$. Another K\"ahler potential
$\hat{f}_0(\eta, \overline{\eta})=f_0-\mbox{ln}{|\zeta|^2}$
is well defined near $\eta=0$ but blows up
near $\zeta=0$. Consider the vector field $W=\zeta\partial_\zeta$ which on the overlap $U\cap\widetilde{U}$ is
equal to $-\eta\partial_{\eta}$. 
Rewriting the definition of $\hat{f}_0$ as
\be
\label{hatf}
\hat{f}_0(\eta, \overline{\eta})-\frac{1}{2}\mbox{ln}{|\eta|^2}  =f_0(\z, \zb) -\frac{1}{2}\mbox{ln}{|\zeta|^2}
\ee 
and Lie deriving both sides of (\ref{hatf}) along $W$ we find
$\zeta (f_0)_{\zeta}-1/2=-(\eta {(\hat{f}_0)}_\eta-1/2)$. 
Choosing the K\"aher potential for $(M, g)$ to be
\[
f= \begin{cases}
  f_0   -\frac{1}{2}\mbox{ln}{|\zeta|^2}  +f_1  & \text{on}\;U \\
  \hat{f}_0 -\frac{1}{2}\mbox{ln}{|\eta|^2} +f_1 & \text{on}\; \widetilde{U}
\end{cases}
\]
we find that it has two singularities, one in $U$ and one in $\widetilde{U}$, but
\be
\label{globality}
\zeta f_{\zeta}-1/2
\ee
is a global and smooth function on $S^2$. 

We will now establish a general local result for the equation (\ref{eq:HEE}) in 2-dimensions.
\begin{prop}
\label{prop1}  Let $(g, X)$ satisfy the system (\ref{eq:HEE}) on a 2-dimensional surface $M$.  
In any holomorphic chart $(\z, \zb)$ there exists a K\"ahler potential $f$ such that
\be
\label{gandV}
g=4f_{\zeta\bar{\zeta}} \; \td\zeta \td\bar{\zeta}, \quad 
\xb=-2 f_{\zeta\bar{\zeta}}
\left( \frac{\td\zeta}{f_{\bar{\zeta}}}+\frac{\td\bar{\zeta}}{f_{{\zeta}}} \right) ,
\ee
and the the system (\ref{eq:HEE}) reduces to a single 4th order PDE for the K\"ahler potential
$f$
\be
\label{pde}
\frac{\p Q}{\p \zb}+\frac{\p \overline{Q}}{\p \zeta}=0, \quad
\mbox{where}\quad Q=({f_\z})^2\frac{\p}{\p \z}\Big(\ln{\frac{f_{\zb} f_{\z\zb}}{f_{\z}}}\Big) + \frac{2\lambda}{3} f_{\z}^3.
\ee
This equation is equivalent to (\ref{vector_pde}).
\end{prop}
\begin{proof} 
Choose holomorphic coordinates and an associated K\"ahler potential so $g$ is of the form (\ref{kahler11}). In these coordinates write
\[
\xb=A\td\zeta+\overline{A}\td\zb,
\]
where $A$ is a smooth complex valued function.
The $(\zeta\zeta)$ and  $(\zb\zb)$
components of (\ref{eq:HEE}) yield
\[
\p_\zeta A-\frac{1}{2}A^2-\frac{f_{\zeta\zeta\zb}}{f_{\zeta\zb}}A=0,
\]
and the complex conjugate of this equation. These two equations can be integrated to find $A$ in terms of $f$ and its derivatives to be
\be
\label{A0}
A=-2\frac{f_{\z\zb}}{\p_{\zb}f+\p_{\zb}\overline{C}},
\ee
where $C=C(\z)$ is an arbitrary holomorphic function. The freedom of adding a holomorphic function and its conjugate to $f$ can be used to set $C$ to zero.
This choice yields
the form of $\xb$ in (\ref{gandV}). 
Finally, using  the formula for the Ricci scalar
\begin{equation}
R=-(\Delta_0 f)^{-1} \Delta_0(\ln{(\Delta_0 f)}), \quad\mbox{where}\quad\Delta_0=4\frac{\p^2}{\p\zeta\p\bar{\zeta}}  \label{eq:flatlap}
\end{equation} 
is the flat Laplacian,
the  $(\zeta\zb)$
component of (\ref{eq:HEE}) gives the 4th order PDE (\ref{pde}) for $f$.

Expanding the derivatives and isolating the numerator in (\ref{pde}) we find that 
this equation is equivalent to  (\ref{vector_pde}) if the gradient and Laplacian
are taken with respect to the Cartesian coordinates 
$(\frac{1}{2}(\z+\zb), \frac{1}{2i}(\z-\zb))$.
\end{proof}

\begin{remark}
A priori it is not obvious that the choice of K\"ahler potential in Proposition \ref{prop1}  is consistent with the globality condition 
(\ref{globality}) on $M=S^2$.  The $g$--norm of $X$ in \eqref{gandV}  gives
\begin{equation}
|X|^2 =\frac{4 f_{\z{\zb}}}{| f_\z|^2}.  \label{eq:Xnorm}
\end{equation}
By the uniformization theorem  $e^u =f_{\z\zb}(1+|\z|^2)^2$ is global as it is the conformal factor in \eqref{guniform}. With that in mind multiply the numerator and denominator of $|X|^2$ by $|\zeta|^2$ which gives
\[
|X|^2=\frac{|\z|^2}{(1+|\z|^2)^{2}} \frac{e^u}{|\z f_\z |^2}.
\]
The first factor is global as $|\z|^2/(1+|\z|^2)^2=|\eta|^2/(1+|\eta|^2)^2$ if $\eta=1/\z$. Thus $\z f_{\z}$ extends to a smooth function everywhere on $S^2$, apart  from perhaps the zeros of $X$.
\end{remark}

Interestingly, the PDE (\ref{pde}) admits a holomorphic first integral since it is equivalent to
\be
\label{Qintegral}
Q_{\zb}=0
\ee
with $Q$ defined as in (\ref{pde}). The expression $Q_{\zb}$ is complex, but putting all terms in the common denominator gives an expression with a real numerator - this is precisely
equation (\ref{pde}) -  and a complex denominator of lower order.
This implies that $Q=S(\z)$ for some holomorphic function $S$. If $S$ is non--zero, let $w$ be a holomorphic
coordinate on $M$ such that $S^{-1/3}\p_{\z}=\p_{w}$ (this has an effect of setting $S=1$). The resulting PDE for $f=f(w)$ is now the 3rd order PDE
\[
\frac{(f_w)^2f_{ww\wb}}{f_{w\wb}}-f_{w}f_{ww}+\frac{(f_{w})^2f_{w\wb}}{f_{\wb}}
+\frac{2\lambda}{3}(f_w)^3
=1 ,
\]
and its complex conjugate.  If $S$ identically vanishes, one instead obtains the 3rd order PDE $Q=0$ and its complex conjugate.  It would be interesting to obtain a general local solution to (\ref{pde}).
\begin{remark}
\label{remark0}
If the function $S$ vanishes identically the resulting system of PDEs  can be solved exactly. The condition
$Q=0$ integrates to
\be
\label{rem01}
\mbox{ln}{\Big(\frac{f_{\zb}f_{\z\zb}}{f_{\z}}\Big)}+\frac{2\lambda}{3}f=\overline{R(\z)}
\ee
for some holomorphic function $R$. If $R\neq 0$, we let $w$ be a holomorphic coordinate such that
$e^{-R/2}\p_{\z}=\p_w$. Taking the real and imaginary parts of (\ref{rem01}) gives
\[
f_{\bar{w}}=\pm f_w, \quad f_{w\bar{w}}=e^{-2\lambda f/3}.
\]
Set $w=x+iy$ and (swapping $x$ and $y$ if necessary) choose the plus sign in the first equation above. 
Then $f=f(x)$ satisfies the Liouville ODE $f_{xx}=4e^{-2\lambda f/3}$, and the Ricci scalar of the corresponding metric $g$ is constant and equal to $2\lambda/3$. There are three cases to consider:
\begin{itemize}
\item If $\lambda=0$ then reabsorbing one constant of integration into $x$ we get $f=2x^2$, and
\[
g=4(\td x^2+\td y^2), \quad \xb=-2 x^{-1}\td x. 
\]
\item If $\lambda\neq 0$ then reabsorbing one constant of integration into $x$ gives the general solution
\[
f=\frac{3}{\lambda}\mbox{ln}{\Big(c\;e^{-\frac{2xc}{3}} +6\lambda\Big)}+\frac{c}{\lambda}x-\frac{18\mbox{ln}(c)+6\mbox{ln}{(2)} }{4\lambda},
\]
where $c$ is a positive constant of integration. The resulting constant curvature metric can be brought into a standard spherical or hyperbolic form.  If $\lambda>0$ set
\[
x=-\frac{3}{2c}\mbox{ln}{\Big(\frac{6\lambda(1+\cos{\theta})}{c(1-\cos{\theta})}\Big)},\quad
y=\frac{3\phi}{c}
\]
which yields
\[
g=\frac{3}{\lambda}\Big(\td \theta^2+\sin^2{\theta} \td \phi^2\Big), \quad \xb=2\mbox{tan}{(\theta)}\td \theta.
\]
If $\lambda<0$ set
\[
x=-\frac{3}{2c}\mbox{ln}{\Big( \frac{6\lambda(1+\cosh{\theta})}{c(1-\cosh{\theta})}\Big)},\quad
y=\frac{3\phi}{c}
\]
which yields
\[
g=-\frac{3}{\lambda}\Big(\td \theta^2+\sinh^2{\theta} \td \phi^2\Big), \quad \xb=-2\mbox{tanh}{(\theta)}\td\theta.
\]
\end{itemize}
Observe that for all these solutions $X$ is a gradient field, and give further examples of solutions where $X$ is not invariant under all the Killing fields of $g$. The $\lambda>0$ solution gives the round metric on $S^2$, although $X$ is singular (this is the $a=0$ solution we encountered earlier in Remark \ref{constR}).
\end{remark}

\begin{remark}
For the Kerr horizon (\ref{kerr_exact}) the first integral $Q$ has a geometric interpretation which is related to the axial isometry. To see it
set $\Omega=\star \td\xb$, and consider a complex valued function
\[
T=\Big(\frac{R+i\Omega}{2}\Big)^{-1/3}.
\]
We then find 
$(f_{\z})^3 T^{-3}=Q$, or 
$T=\frac{\p f}{\p w}
$
and one can verify that
\[
K= \Big(\frac{2i}{a^2}\Big)^{1/3} \star \td T.
\]
The analogous relation also holds 
for isolated non--extremal horizons \cite{LP}.
\end{remark}

\subsection{Axial symmetry}

We shall now show that if $g$ is assumed to be axisymmetric  
on $M=S^2$ the 4th order ODE resulting from (\ref{pde}) is solvable by quadratures.   We emphasise that Theorem \ref{maintheo0} guarantees the existence of a Killing  field $K$ on $M=S^2$ that  preserves $X$ (for non-gradient $X$).

If $g$ is invariant under an isometric
$U(1)$ action, then we can adapt the coordinates to this action so that
\[
g=\frac{e^u}{(1+r^2)^2}(\td r^2+r^2\td \phi^2), \quad\mbox{where}\quad u=u(r) ,
\]
where $K=\partial_\phi$ is the Killing field with $2\pi$-periodic orbits.
Now choosing the holomorphic coordinate $\zeta=re^{i\phi}$ (so that the $U(1)$-action is $\zeta\rightarrow e^{i\psi}\zeta$), we see that the metric takes the form (\ref{kahler11}) and the K\"ahler potential satisfies
\be
\label{Kandf}
K(f_{\zeta\zb})=0, \quad\mbox{where}\quad
K=i(\zeta\p_\zeta-\zb\p_{\zb}).
\ee
The next result establishes a relation between the Killing vector $K$ of Theorem \ref{maintheo0} and the K\"ahler potential in Proposition \ref{prop1}.

\begin{prop}\label{prop:inherit}
If $(g, X)$ satisfy \eqref{eq:HEE} on $M=S^2$,
and $K$ is a Killing vector for $g$ as in Theorem \ref{maintheo0}, then $K(f)$ is a constant where $f$ is a K\"ahler potential as in Proposition \ref{prop1}.
\end{prop}

\begin{proof}
By Theorem \ref{maintheo0} we have the  invariance condition $[K, X]=0$.
Using (\ref{eq:Xnorm}) one can rewrite $\xb$ in (\ref{gandV}) in the compact form\footnote{Thus $\td f$ where $f$ is our K\"ahler potential is  Jezierski's variable $\Phi$~\cite{Jezierski}.}
\begin{equation}
\xb= -\tfrac{1}{2} |X|^2 \td f .  \label{eq:Xf}
\end{equation} 
Therefore, $\mathcal{L}_K X=0$ implies invariance of  $\td f$ along a Killing field $K$.  
Conversely, we can invert the previous relation so $X^\flat= -2 \td f/|\td f|^2$, which implies that if $\td f$ is Lie-derived by a Killing field $K$ then so is $X^\flat$.  We conclude that  $[K, X]=0$ iff ${\mathcal L}_K(\td f)=\td (K(f))=0$ or equivalently if $K(f)$ is a constant. 
\end{proof}

\begin{remark}  
If we simply {\it assume} that $(M,g)$ admits a Killing field $K$  then the above gives a purely local  proof that $[K,X]=0$  is equivalent to $K(f)$ being a constant.
\end{remark}

We shall now impose the symmetry condition on the 4th order PDE (\ref{pde}), and find its general solution on $M=S^2$.
\begin{prop}
\label{propS}
If $(M, g)$ from Proposition \ref{prop1} admits a Killing vector $K=i(\zeta\p_\zeta-\zb\p_{\zb})$
then there exist coordinates $(x, \phi)$ on $M$, and a function $B=B(x)$ such that
\be
\label{B_met}
g=B^{-1}dx^2+Bd\phi^2, \quad \xb=-\frac{1}{2}\Big(\frac{dx+iBd\phi}{x/2+\ov{H}}+
\frac{dx-iBd\phi}{x/2+{H}}\Big),
\ee
where $H=H(\zeta)$ is holomorphic. Moreover if $M=S^2$ and $X$ is non-gradient,  
then $H$ is a constant, and 
the 4th order PDE (\ref{pde}) reduces to the linear ODE
\be
\label{b_eq}
B''+\frac{2}{(\beta^2+x^2)^2}\Big(x(\beta^2+x^2)B'+2\beta^2B\Big)+2\lambda=0,
\ee
where $\beta\neq 0$ is a  constant, 
which can be solved explicitly, and $X$ is given by
\begin{equation}
\xb=  \frac{2\beta B(x)\td\phi}{\beta^2+x^2}-  \frac{2x\td x}{\beta^2+x^2} . \label{eq:X2d}
\end{equation}
\end{prop}
 \begin{proof} 
 The Killing equation ${\mathcal L}_Kg=0$ with $K$ as in (\ref{Kandf}) gives $\zb f_{\zeta\zb\zb}-\z f_{\zeta\zb \zeta}=0$.
The solution to this is
\[
f=\rho(r)+A(\zeta)+\bar{A}(\zb), \quad \mbox{where}\quad r=|\zeta|, 
\]
and $\rho(r)$, $A(\z)$, $\bar{A}(\zb)$ are integration functions. 
Now set $\zeta=e^{s+i\phi}$, where $s$ and $\phi$ are real. Moreover set $x=d \rho/d s$, and find that
\[
g=\frac{\td x}{\td s}\Big(\td s^2+\td\phi^2), \quad
\xb=-\frac{1}{2}\frac{\td x}{\td s}\Big(\frac{\td s+i\td \phi}{x/2+\zb A_{\zb}} + \frac{\td s-i\td \phi}{x/2+\zeta 
A_{\zeta}}\Big).
\]
Motivated by Sylvester's theory of reciprocants \cite{sylvester}, 
swap the dependent and independent variables, and regard $s$ as a function 
of $x$. Setting 
\be
\label{def_of_B}
B=\Big(\frac{\td s}{\td x}\Big)^{-1}
\ee
puts $(g, X)$  in the form (\ref{B_met}),
where $H(\zeta)=\zeta A_{\zeta}$. 

Now if $M=S^2$ and $X$ is non-gradient, we may take $K$ to be the Killing field guaranteed by Theorem \ref{maintheo0}, and Proposition \ref{prop:inherit} shows that $K(f)$ is a constant. From above we have $\z f_\z =\tfrac{1}{2}x+\z A_\z$ and so $K(f)= i (\z A_\z - \zb \bar{A}_{\zb})$ is a constant.
Thus
$\zeta A_{\zeta}$ is both holomorphic and antiholomorphic so must be  
a constant.  Set $H=\zeta A_{\zeta}=(\alpha+i\beta)/2$, where $\alpha, \beta$ are real constants. Now 
\[
\xb=\frac{-2}{(x+\alpha)^2+\beta^2}\Big((x+\alpha) \td x-B\beta \td\phi\Big).
\]
The constant $\alpha$ may be set to $0$ by translating the function $x$ giving (\ref{eq:X2d}). Note that $X^\flat$ is exact iff $\beta=0$.
The PDE (\ref{pde})
reduces to an ODE 
which linearises under the reciprocal transformation (\ref{def_of_B}). The resulting equation
for $B=B(x)$ is (\ref{b_eq}).
\end{proof}

\begin{remark}
The above proof reveals that $\z f_\z$ is a global function on $S^2$ inline with (\ref{globality}). This follows from the fact that $x$ is globally defined which can be seen as follows. The 1-form $\iota_{K} \epsilon_g$, where  $\epsilon_g$ is the volume form of $g$, must be closed and hence exact on $M=S^2$, see e.g.~\cite{Kunduri:2008rs}. Computing this using (\ref{B_met}) gives $\td x  = -\iota_{K} \epsilon_g$ and thus we deduce $x$ is indeed global.
\end{remark}

The general solution to (\ref{b_eq}) is 
\[
B(x)=\frac{b_1(\beta^2-x^2)+b_2 \beta x -\frac{\Lambda}{3}(x^4+3\beta^4)}{\beta^2+x^2}
\]
if $\beta\neq 0$, or  $B=b_1+b_2/x-\Lambda x^2/3$ if $\beta=0$. The latter corresponds to $\xb$ being exact and hence this solution does not extend to a smooth metric on a compact $M$ (it is in fact the general static solution)~\cite{Chrusciel:2005pa}. The $\beta\neq 0$ solution with 
$\Lambda=0, b_1=2, b_2=0, \beta=1$ gives the extremal Kerr near-horizon geometry (\ref{kerr_exact}) with $a^2=1/2$. Moreover, the $\beta \neq 0$ solution can be shown to extend to a smooth metric on $S^2$ iff it is isometric to the extreme Kerr horizon  (\ref{kerr_exact})  or its counterpart with a cosmological constant~\cite{Lewandowski:2002ua, Kunduri:2008rs, Kunduri:2013gce}.    \\

\noindent {\bf Acknowledgements.}  We are grateful to Alex Colling, David Katona and Paul Tod for discussions. 
JL acknowledges support by the Leverhulme Trust Research Project Grant RPG-2019-355. We thank the anonymous referees for comments that led to improvements and in particular to a corrected version of Theorem \ref{cor2}.

\appendix

\section{Details of proof of Proposition \ref{prop}}
\label{app:mainidentity}

Here we give more details on the proof of our main identity \eqref{eq:identity}.  The trace of \eqref{eq:nhE2} gives the scalar curvature
\be
R=  S + n\lambda  \; ,
\ee
hence the Einstein tensor is
\be
R_{ab}-\tfrac{1}{2} R g_{ab} =  \frac{K_a K_b }{2\Gamma^2} -\frac{(\nabla_a \Gamma)( \nabla_b\Gamma)}{2 \Gamma^2}- \frac{1}{\Gamma}\nabla_{(a} K_{b)}+\frac{1}{\Gamma} \nabla_a\nabla_b \Gamma+\lambda \left(1-\frac{n}{2} \right)  g_{ab}  - \tfrac{1}{2} g_{ab} S  \; , \label{eq:contbianchi}
\ee
where for convenience we have defined the function
\be
S:=  \frac{1}{2\Gamma^2} |K|^2 - \frac{1}{2\Gamma^2} |\nabla \Gamma|^2-\frac{1}{\Gamma} \nabla_c K^c  +\frac{1}{\Gamma} \Delta \Gamma    \; . \label{eq:Gdef}
\ee
The contracted Bianchi identity $\nabla^a(R_{ab}-\tfrac{1}{2} g_{ab}R)=0$ implies that the r.h.s. of \eqref{eq:contbianchi} is divergence-free and this condition can be rearranged to give an equation for $\nabla^a\nabla_{(a} K_{b)}$ which is \eqref{eq:calc1}.

The term  on the r.h.s. of \eqref{eq:calc1} that is a triple derivative in $\Gamma$ can be evaluated using $[\nabla^a,\nabla_b] \nabla_a \Gamma = R_{ab} \nabla^a \Gamma$ and \eqref{eq:nhE2} which gives
\begin{align*}
\nabla^a\nabla_b \nabla_a \Gamma  &=  \frac{1}{2\Gamma^2} (K^a\nabla_a \Gamma) K_b- \frac{1}{2\Gamma^2} | \nabla \Gamma |^2 \nabla_b \Gamma -\frac{1}{\Gamma} (\nabla^a\Gamma) \nabla_{(a} K_{b)} \\ &+ \frac{1}{\Gamma}( \nabla^a\Gamma) \nabla_a \nabla_b\Gamma +\nabla_b \Delta \Gamma+\lambda \nabla_b \Gamma
\end{align*}
and substituting this into \eqref{eq:calc1} one finds a few cancellations leaving
\begin{align}
 \nabla^a \nabla_{(a} K_{b)}  &=\frac{1}{2\Gamma} K_b \nabla_c K^c +\frac{1}{2\Gamma} K^a \nabla_a K_b  - \frac{1}{2\Gamma^2} K_b K^c \nabla_c\Gamma - \frac{1}{2\Gamma} (\Delta \Gamma)  \nabla_b \Gamma  \label{eq:calc2}  \\  \nonumber
 &-\frac{1}{2\Gamma} (\nabla^a\Gamma)\nabla_a \nabla_b\Gamma  + \frac{1}{2\Gamma^2} | \nabla \Gamma|^2 \nabla_b \Gamma  +\nabla_b \Delta \Gamma+\lambda \nabla_b \Gamma -\tfrac{1}{2} \Gamma \nabla_b S\;.
\end{align}
Now contracting this with $K^b$ gives
\begin{align}
K^b \nabla^a \nabla_{(a} K_{b)}  & =K^a \nabla_a\Gamma \left(-\frac{1}{2\Gamma^2} |K|^2  -\frac{1}{2\Gamma} \Delta \Gamma+\frac{1}{2\Gamma^2} |\nabla \Gamma |^2+ \lambda \right) + \frac{1}{2\Gamma} |K|^2 \nabla_c K^c \label{eq:calc3}  \\
& + \frac{1}{4\Gamma} K^a\nabla_a |K|^2 -\frac{1}{4\Gamma} K^b\nabla_b |\nabla \Gamma|^2  + K^b\nabla_b \Delta \Gamma - \tfrac{1}{2} \Gamma K^b \nabla_b S \; . \nonumber  
\end{align}
We can rewrite each term on the r.h.s. of \eqref{eq:calc3} as a sum of a total divergence, a term proportional to $\nabla_aK^a$ and a term proportional to $K^a \nabla_a\Gamma$. The terms on first line are already of this form. Now consider the second line of \eqref{eq:calc3}: the first term  can be rewritten as
\begin{align*}
\frac{1}{4\Gamma} K^a\nabla_a |K|^2= \nabla^a \left( \frac{1}{4\Gamma} K_a |K|^2 \right) - \frac{1}{4\Gamma}|K|^2 \nabla_a K^a+ \frac{1}{4\Gamma^2} |K|^2 K^a\nabla_a\Gamma  \; ,
\end{align*}
the second term as
\begin{align*}
-\frac{1}{4\Gamma} K^b\nabla_b |\nabla \Gamma|^2 =  -\nabla^b \left( \frac{1}{4\Gamma} K_b |\nabla \Gamma|^2\right)+  \frac{1}{4\Gamma} (\nabla_b K^b) |\nabla \Gamma|^2 -  \frac{1}{4\Gamma^2} (K^b\nabla_b\Gamma) |\nabla \Gamma|^2  \; ,
\end{align*}
the third term as
\begin{align*}
K^b\nabla_b \Delta \Gamma = \nabla^b( K_b \Delta \Gamma)- (\nabla_bK^b) \Delta \Gamma
\end{align*}
and the fourth term as
\begin{align*}
 - \tfrac{1}{2} \Gamma K^b \nabla_bS  = -\nabla^b \left( \tfrac{1}{2} \Gamma K_b S \right)+ \tfrac{1}{2} S \nabla_b(\Gamma K^b)   \; .
\end{align*}  
Substituting each of these into \eqref{eq:calc3} and collecting like terms gives
\begin{align*}
K^b \nabla^a \nabla_{(a} K_{b)}  & =K^a \nabla_a\Gamma \left(-\frac{1}{4\Gamma^2} |K|^2  -\frac{1}{2\Gamma} \Delta \Gamma+\frac{1}{4\Gamma^2} |\nabla \Gamma |^2+ \lambda +\tfrac{1}{2}S \right) \\
&  + \nabla_cK^c \left(  \frac{1}{4\Gamma} |K|^2  +  \frac{1}{4\Gamma}  |\nabla \Gamma|^2 -\Delta \Gamma +\tfrac{1}{2}\Gamma S \right)  \nonumber \\
& + \nabla^b \left( \frac{1}{4\Gamma} K_b |K|^2-  \frac{1}{4\Gamma} K_b |\nabla \Gamma|^2  + K_b \Delta \Gamma- \tfrac{1}{2} \Gamma K_b S \right)   \\
&= K^a \nabla_a\Gamma   \left( \lambda -\frac{1}{2\Gamma} \nabla_a K^a \right)  +\nabla_c K^c \left( \frac{1}{2\Gamma}|K|^2-\frac{1}{2}\Delta \Gamma- \frac{1}{2} \nabla_aK^a \right) \\ &+ \nabla^b  \left( \tfrac{1}{2}  K_b \Delta \Gamma+\tfrac{1}{2}K_b \nabla_a K^a \right) \; ,
\end{align*}
where the second equality follows from using \eqref{eq:Gdef}.  Notice that multiple cancellations have occurred, particularly for the terms proportional to $K^a\nabla_a \Gamma$.  In fact, we can get rid of the terms proportional to $K^a\nabla_a \Gamma$ completely by rewriting the term proportional to the cosmological constant using $K^a\nabla_a\Gamma= \nabla^a(K_a \Gamma)-\Gamma \nabla_aK^a$.  This results in
\begin{align}
K^b \nabla^a \nabla_{(a} K_{b)}  & =\nabla_c K^c \left(-\lambda \Gamma -\frac{1}{2\Gamma} K^a\nabla_a\Gamma+ \frac{1}{2\Gamma}|K|^2-\frac{1}{2}\Delta \Gamma- \frac{1}{2} \nabla_aK^a \right) \\ &+ \nabla^b  \left( \lambda \Gamma K_b+ \tfrac{1}{2}  K_b \Delta \Gamma+\tfrac{1}{2}K_b \nabla_a K^a \right) \; .  \nonumber
\end{align}
Finally, substituting this into \eqref{eq:Killingeqsq} gives the claimed identity \eqref{eq:identity}.

\section{Details of proof of Theorem \ref{prop2d}}
\label{app:id2d}

Here we give more details on the proof of the identity \eqref{eq:id2d}.  Since $K^a$ is a Killing field and the dimension $n=2$ we can write \eqref{eq:nhE2} as
\be
\Gamma \left( \tfrac{1}{2} R-\lambda \right)g_{ab}= \frac{K_a K_b}{2\Gamma}- \frac{(\nabla_a \Gamma)( \nabla_b \Gamma)}{2\Gamma}+ \nabla_a\nabla_b \Gamma  \;.  \label{eq:nhE3}
\ee
Now acting on both sides with $\epsilon^{ac}\nabla_c$ gives
\begin{align}
\epsilon_{bc}\nabla^c \left( \Gamma \left( \tfrac{1}{2} R-\lambda \right) \right) = \epsilon^{ac}\nabla_c \left( \frac{K_a K_b}{2\Gamma}\right)  - \frac{1}{2\Gamma} \epsilon^{ac}( \nabla_a\Gamma )\nabla_c\nabla_b\Gamma + \epsilon^{ac} \nabla_c \nabla_a \nabla_b\Gamma  \; .  \label{eq:calc4}
\end{align}
The term with a triple derivative in $\Gamma$ evaluates to
\begin{align*}
\epsilon^{ac} \nabla_c \nabla_a \nabla_b\Gamma = \tfrac{1}{2}\epsilon^{ac} [\nabla_c, \nabla_a ]\nabla_b\Gamma = -  \tfrac{1}{2}\epsilon^{ac}R^d_{~bca}\nabla_d\Gamma = - \frac{1}{2}R \epsilon_{bc}\nabla^c \Gamma
\end{align*}
where in the final equality we used the 2-dimensional identity $R_{dbca}= \tfrac{1}{2} (g_{dc} g_{ba}-g_{da} g_{bc})$.  Now return to \eqref{eq:calc4}, substitute for the triple derivative in $\Gamma$ and substitute for $\nabla_b \nabla_c\Gamma$ using \eqref{eq:nhE3}, to get
\begin{align*}
\Gamma \epsilon_{bc}\nabla^c  \left( \tfrac{1}{2} R-\lambda \right) + \tfrac{1}{2} \epsilon_{bc} (\nabla^c \Gamma) \left( \tfrac{3}{2}R-\lambda\right) = \epsilon^{ac}\nabla_c \left( \frac{K_a K_b}{2\Gamma}\right)+ \frac{1}{4\Gamma^2} \epsilon^{ac} (\nabla_a \Gamma) K_c K_b   \; .  \label{eq:calc5}
\end{align*}
Now contract this with $\epsilon^b_{~d}$, use $\epsilon^{bd}\epsilon^{ac}=g^{ba}g^{dc}-g^{bc}g^{da}$ and simplify to find
\begin{align*}
\frac{2}{3}\Gamma \nabla_d R+  \left( R-\frac{2\lambda}{3} \right)\nabla_d\Gamma = \nabla_d\left( \frac{|K|^2}{\Gamma} \right)+ \frac{1}{\Gamma^2} K_d K^c \nabla_c \Gamma  \; .
\end{align*}
The r.h.s of this is $2 K^b \nabla_{[d} X_{b]}$, where recall $X_b$ is related to $K_b$ by \eqref{eq:Xdecomp}, and hence we obtain \eqref{eq:id2d}.

\end{document}